\documentclass[a4paper,USenglish]{article}
\usepackage[linesnumbered,ruled,vlined]{algorithm2e}
\usepackage[ruled,linesnumbered]{algorithm2e}
\SetKwRepeat{Do}{do}{while}
\usepackage{amsmath,amsfonts,amsthm,amssymb,color}
\usepackage{latexsym,bbm,xspace,graphicx,float}
\usepackage{cite}
\usepackage{hyperref}
\usepackage[capitalise, nameinlink]{cleveref}
\usepackage{cs270}
\SetKwRepeat{Do}{do}{while}
\pdfoutput=1
\usepackage[normalem]{ulem}

\title{Deterministic Minimum Steiner Cut in Maximum Flow Time}
\author{Matthew Ding\\University of California, Berkeley\\Email: \tt{matthewding@berkeley.edu}
\and
Jason Li\\Carnegie Mellon University\footnote{This work was done while visiting the Simons Institute for the Theory of Computing.}\\Email: \tt{jmli@cs.cmu.edu}}
\date{\today}

\begin{document}
\maketitle
\begin{abstract}
    We devise a deterministic algorithm for minimum Steiner cut, which uses $(\log n)^{O(1)}$ maximum flow calls and additional near-linear time. This algorithm improves on Li and Panigrahi's (FOCS 2020) algorithm, which uses $(\log n)^{O(1/\epsilon^4)}$ maximum flow calls and additional $O(m^{1+\epsilon})$ time, for $\epsilon > 0$. Our algorithm thus shows that deterministic minimum Steiner cut can be solved in maximum flow time up to polylogarithmic factors, given any black-box deterministic maximum flow algorithm. Our main technical contribution is a novel deterministic graph decomposition method for terminal vertices that generalizes all existing $s$-strong partitioning methods, which we believe may have future applications.
\end{abstract}

\section{Introduction}
The minimum cut (or ``min-cut'') of a weighted graph is the smallest weighted subset of edges whose deletion disconnects the graph. The problem of finding the minimum cut is one of the most fundamental problems in combinatorial optimization and theoretical computer science as a whole. It also has important applications such as network optimization \cite{AHUJA19951} and image segmentation \cite{BPHN_2019}. Thus, finding faster algorithms for this problem will have far-reaching applications for a wide variety of fields. Recently, there has been a large amount of groundbreaking work in the field, including deterministic almost-linear time\footnote{Near-linear algorithms have runtime $\Tilde{O}(m)$ and almost-linear algorithms have runtime $m^{1+o(1)}$. We use $\tilde{O}(\cdot)$ to hide polylogarithmic factors.} algorithms for both minimum cut \cite{li2021deterministic} and maximum flow \cite{determinstic_max_flow}.  

\subsection{Minimum Steiner Cut Background}

A classic extension of the min-cut problem is the minimum Steiner cut (or ``Steiner min-cut'') problem. In this problem, we are given an undirected, weighted graph $G = (V,E)$ and a subset $T \subseteq V$ of terminals. A Steiner cut is a subset of edges whose removal disconnects at least one pair of terminals in the graph. The minimum Steiner cut is the Steiner cut with the minimum total weight of cut edges. This problem generalizes both $s-t$ minimum cut ($T = \{s,t\}$) and global minimum cut ($T=V$) and is therefore a fundamental problem in graph algorithms.

The classical algorithm to solve minimum Steiner cut uses $|T| - 1$ max-flow computations. Li and Panigrahi \cite{steiner-polylog-flows} give a randomized algorithm which reduces minimum Steiner cut in near-linear time to just polylogarithmic number of max-flow computations. They additionally provide a deterministic algorithm which takes, for any parameter $\epsilon>0$, $(\log n)^{O(1/\epsilon^4)}$ max-flow calls with $O(m^{1+\epsilon})$ additional running time. Given the currently known fastest deterministic maximum flow algorithm in almost-linear time \cite{determinstic_max_flow}, the two results combined give an almost-linear time algorithm for global minimum cut, which matches the algorithm of Li \cite{li2021deterministic}. We remark that a very recent work~\cite{near_linear_mincut} has improved the running time of deterministic \emph{global} minimum cut to near-linear, i.e., \ $\tilde O(m)$. However, since minimum Steiner cut is at least as hard as $s-t$ minimum cut, traditionally solved through $s-t$ max-flow, a near-linear time minimum Steiner cut algorithm remains elusive without an equally fast max-flow algorithm.

\subsection{Our Contributions}
We show that a deterministic, near-linear time max-flow algorithm is the \emph{only} obstacle towards obtaining a deterministic, near-linear time minimum Steiner cut algorithm. More precisely, we introduce a new deterministic algorithm that finds the minimum Steiner cut in polylogarithmic $s-t$ max flow calls and near-linear additional processing time.
\begin{theorem}\label{theorem:main_result}
    Given an undirected, weighted graph $G=(V,E)$ with $n$ vertices and $m$ edges, polynomially bounded edge weights, and a set of terminal vertices $T \subseteq V$, there is a deterministic minimum Steiner cut algorithm that makes $\text{polylog}(n)$ maximum flow calls on undirected, weighted graphs with $O(n)$ vertices and $O(m)$ edges, and runs in $\Tilde{O}(m)$ time outside of these maximum flow calls.
\end{theorem}

Specifically, a hypothetical \textbf{deterministic near-linear time algorithm for $s-t$ max-flow also implies a deterministic near-linear time algorithm for minimum Steiner cut}. This result was not known from the work of \cite{steiner-polylog-flows}, given the additional $m^{1+\epsilon}$ running time in their deterministic algorithm. In fact, for the case of polylogarithmic maximum flow calls ($\epsilon=\Omega(1)$), their algorithm has runtime $m^{1+\Omega(1)}$, even slower than almost-linear. Our algorithm also matches the running time of the randomized minimum Steiner cut algorithm from \cite{steiner-polylog-flows} up to polylogarithmic factors.

\paragraph*{Terminal-Based Partitioning Methods.}
Expander decompositions have been a powerful tool for solving minimum cut problems in recent years \cite{li2021deterministic,steiner-polylog-flows}. However, the current state-of-the-art deterministic expander decomposition takes almost-linear time \cite{almost_linear_expander}, and it is an open problem whether this can be improved.

A key tool for deterministic near-linear time algorithms is the decomposition into $s$-strong clusters used by recent minimum cut algorithms on \emph{simple} graphs \cite{HRW20, edge-connectivity}. However, recent work \cite{near_linear_mincut} has applied the concept of $s$-strong clusters to weighted graphs to obtain a deterministic near-linear global minimum cut algorithm. Our main technical contribution is to extend this decomposition framework on general weighted graphs and apply it towards a decomposition that \emph{specifically partitions clusters of terminals with a boundary proportional to the size of the terminal set}. We state this result informally below.
\begin{theorem}[Informal]\label{theorem:terminal_decomposition}
    Given an undirected, weighted graph $G=(V,E)$ with $n$ vertices and $m$ edges, polynomially bounded edge weights, a set of terminal vertices $T \subseteq V$, sparsity parameter $0 < \psi < 1$, and cut size parameter $\delta > 0$, there is a deterministic algorithm which returns a vertex partitioning of clusters $V_1, V_2,...,V_\ell$ such that the following hold:
    \begin{enumerate}
        \item For every cluster, any cut with weight less than $\delta$ splits the cluster with at most $\textup{polylog(n)}$ terminals on at least one side. 
        \item For every cluster, any cut with weight less than $\delta$ either does not split the cluster or has at least $\geq \delta/\textup{polylog(n)}$ weight of cut edges inside the cluster.
        \item The total weight of edges between clusters is at most $\Tilde{O}(\delta\cdot|T|)$.
    \end{enumerate}
     
    The algorithm makes $O(\log^2 n)$ maximum flow calls on undirected, weighted graphs with $O(n)$ vertices and $O(m)$ edges and runs in $\Tilde{O}(m)$ time outside of these maximum flow calls.
\end{theorem}
Properties 1 and 3 together directly generalize the notion of $s$-strong partitions to the terminal regime. Property~2 provides an additional guarantee useful for the Steiner algorithm of \cite{steiner-polylog-flows} and can be achieved with only polylogarithmic loss elsewhere in the decomposition.

\section{Preliminaries}
In this paper, all graphs are undirected and weighted. For simplicity, all weights are assumed to be polynomially bounded (however, this can be relaxed by adding logarithmic dependence on the maximum edge weight).

We begin by introducing standard definitions and tools from previous works that we will utilize for our algorithm. We also define our new modification of $s$-strong clusters to terminals. 

We use a standard definition of induced subgraphs using self-loops for boundary edges, which preserves degrees of vertices in subgraphs. We denote the induced subgraph of the vertex set $S \subset V$ on graph $G(V, E)$ as $G[S]$.

\subsection{Sparsity and Strength}
Sparsity is a specific measure of how connected a graph is. For a cut $(U, \overline{U})$, where $\overline{U} = V \setminus U$, we define $\partial U = \partial \overline{U}$ as the boundary of the cut, which is the set of edges between $U$ and $\overline{U}$.

\begin{definition}[Sparsity]
    The sparsity of a cut $(U,\overline{U})$ is defined as
\begin{equation}
    \Psi(U) = \frac{w(\partial U)}{\min\{|U|, |\overline{U}|\}} = \Psi(\overline{U})
\end{equation}
\end{definition}

We also introduce the new definition of terminal-sparsity for Steiner cuts with the terminal set $T$.
\begin{definition}[Terminal-sparsity]
The terminal-sparsity of a cut $(U,\overline{U})$ is defined as
\begin{equation}
    \Psi_T(U) = \frac{w(\partial U)}{\min\{|U \cap T|, |\overline{U} \cap T|\}} = \Psi_T(\overline{U})
\end{equation}
\end{definition}
We use the terms \emph{$\psi$-sparse} and \emph{$\psi$-terminal-sparse} to refer to cuts with sparsity and terminal-sparsity $< \psi$, respectively.

The concept of strength, introduced by Kawarabayashi and Thorup \cite{edge-connectivity}, is a relaxed notion of edge expanders. At a high level, a vertex subset $U\subseteq V$ is $s$-strong if every cut $(C,\overline C)$ of weight at most $\delta$ satisfies $\min\{\textbf{vol}(C\cap U),\textbf{vol}(\overline C\cap U)\}\le s$, where the volume $\textbf{vol}(S)$ is the sum of degrees of vertices in $S$. In \cite{edge-connectivity}, the parameter $\delta$ is chosen as the minimum degree of all vertices, which serves as an upper bound on the min-cut.

One of our main conceptual contributions is translating $s$-strength to the terminal setting and providing necessary generalizations to handle the Steiner min-cut problem. First, we work from a \emph{sparsity} viewpoint, which bounds the minimum cardinality of intersection $\min{\{|C \cap U|, |\overline{C} \cap U|\}}$ instead of volume, which is more handy when we start introducing terminals. Second, we can no longer choose $\delta$ as the minimum degree since it no longer upper bounds the Steiner min-cut. One natural choice is the minimum (weighted) degree of all vertices in set $S$, but instead, for technical reasons, we set $\delta$ closer to the minimum Steiner cut $\lambda$ itself. For now, we keep $\delta$ as a free parameter and provide our $s$-strong guarantees in terms of $\delta$. Finally, we need an additional requirement that if the cut $(C,\overline C)$ cuts any edges inside a cluster, then it must cut sufficiently many such edges, and we introduce another parameter $\gamma$ to capture this condition.

\begin{definition}[$(s, \delta, \gamma)$-strength]
    A vertex subset $U \subseteq V$ (called a \emph{cluster}) is \emph{$(s, \delta, \gamma)$-strong in $G$} if every cut $(C, \overline{C})$ of graph $G$ with at most weight $\delta$ satisfies $\min{\{|C \cap U|, |\overline{C} \cap U|\}} \leq s$, and moreover, if $\min{\{|C \cap U|, |\overline{C} \cap U|\}}>0$ then $w(\partial_{G[U]}C) \geq \gamma \cdot \delta$.
\end{definition}

Next, we introduce the notion of $(s, \delta, \gamma)$-terminal-strength, where the ``size'' of a set of vertices is only determined by its number of terminals. This property is necessary to deal with cuts separating terminal vertices instead of just regular ones.

\begin{definition}[$(s, \delta, \gamma, T)$-terminal-strength]
    A vertex subset $U \subseteq V$ (called a cluster) is $(s, \delta, \gamma, T)$-terminal-strong in $G$ if every Steiner cut $(C, \overline{C})$ of graph $G$ with at most weight $\delta$ satisfies \\
    $\min{\{|C \cap U \cap T|, |\overline{C} \cap U \cap T|\}} \leq s$, and moreover, if $|C \cap U \cap T| > 0$ and $|\overline{C} \cap U \cap T|\} > 0$, then $w(\partial_{G[U]}C) \geq \gamma \cdot \delta$. 

    For the rest of the paper, we sometimes omit the ``\emph{in $G$}'' and the terminal set $T$ from the definition whenever they are apparent from the context.
\end{definition}
An important property of both $(s, \delta, 0)$-strength and terminal strength is that the property is inherited by subgraphs, i.e., \ if $G(V,E)$ is $(s, \delta, 0)$-strong or terminal-strong, then $G[A]$ is as well for all $A \subseteq V$. This property holds for $s$-strength \cite{edge-connectivity}, and is straightforward to verify that the same is true for our $(s, \delta, 0)$-strength definitions.

Lastly, we also define terminal-strong decompositions, which are analogous to $s$-strong and expander decompositions, except that we split our graph into $(s, \delta, \gamma)$-terminal-strong components as opposed to $s$-strong sets and expanders, respectively.

\begin{definition}[$(s, \delta, \gamma, T)$-terminal-strong decomposition]
    A set of disjoint vertex clusters \linebreak $V_1, V_2,...,V_\ell$ is an $(s, \delta, \gamma, T)$-terminal-strong decomposition if each cluster $V_i$ is $(s, \delta, \gamma, T)$-terminal-strong, and if the total weight of edges between clusters is at most $\Tilde{O}(\delta\cdot|T|)$.
\end{definition}

Our primary new technical tool is a fast algorithm for computing an $(s,\delta,\gamma,T)$-terminal-strong decomposition with a small bounded weight of intercluster edges (\cref{theorem:terminal_decomposition}), which is presented in detail in \cref{sec:cut_game}. At a high level, we use a (non-terminal) $(s, \delta, \gamma)$-strong decomposition to devise an algorithm that finds an $(s, \delta, \gamma, T)$-terminal-strong decomposition through the \emph{cut-matching game} framework of Khandekar, Rao, and Vazirani~\cite{single-commodity-flows}, which we outline in the following subsection.

Finally, given such a decomposition, we use the framework of \cite{steiner-polylog-flows}, replacing their expander decomposition step with our $(s,\delta,\gamma)$-strong decomposition. We leave the details to \Cref{sec:sparsification}.

\subsection{Cut-Matching Game}
We start with an overview of the cut-matching game.
\begin{enumerate}
    \item The cut player chooses a bisection $(S, \overline{S})$ of the graph $H_{t-1}$ based on a given strategy.
    
    \item The matching player chooses a perfect matching of the bisection based on a given strategy.
    
    \item The cut player adds the edges of the perfect matching to graph $H_{t-1}$, forming graph $H_t$.
\end{enumerate}

The game continues until graph $H_t$ is an edge-expander. The key insight of the cut-matching game is that there is always a strategy for the cut player that finishes the game in few rounds. 

We use the cut-matching game to reduce our problem from one with terminals ($(s, \delta, \gamma)$-terminal-strong decomposition) to one without terminals ($(s, \delta, \gamma)$-strong decomposition). We then adapt the $(s, \delta, 0)$-strong decomposition algorithm of \cite{near_linear_mincut} to obtain an $(s, \delta, \gamma)$-strong decomposition for large enough $\gamma$.

\section{Minimum Steiner Cut Algorithm Overview}
The following is an overview of our algorithm (\cref{alg:steiner-min-cut}) to solve minimum Steiner cut on an undirected, weighted graph $G$ deterministically in near-linear time (i.e., $\Tilde{O}(m)$) plus polylogarithmic maximum flow calls. Throughout, we assume that we have guessed the value of the Steiner mincut up to factor $2$ (which we denote $\tilde\lambda$) by, for example, guessing all powers of $2$ (incorrect guesses may return an overestimate of the minimum Steiner cut, but we can take the minimum cut ever found at the end).
\begin{enumerate}
    \item We use the ``unbalanced case'' of \cite{steiner-polylog-flows} to find the Steiner minimum cut $(C,\overline C)$ if $\min\{|C\cap T|,|\overline C\cap T|\}\le\textup{polylog}(n)$. This algorithm is described in \cref{lemma:unbalanced_case}.
    
    \item In the ``balanced case'', we find an $(s, \delta, \gamma)$-terminal-strong decomposition on the graph. To do this, we use the cut-matching game on a graph $H$ containing only the terminals of the original graph, with $s=\textup{polylog}(n)$, $\delta=\tilde\lambda$, and $\gamma=1/\textup{polylog}(n)$. This algorithm is described in \Cref{alg:cut-game} (\cref{sec:cut_game}). At a high level, we use a (non-terminal) $(s', \delta', \gamma')$-strong decomposition to find an $(s, \delta, \gamma)$-terminal-strong decomposition for appropriate parameters $s',\delta',\gamma',s,\delta,\gamma$.

    \item Using our $(s, \delta, \gamma)$-terminal-strong decomposition, we find a set $T' \subseteq T$ and $|T'| \leq |T|/2$ such that the minimum Steiner cut of $G$ with terminal set $T'$ is the same as with terminal set $T$. In this case, we recursively apply our minimum Steiner cut algorithm on graph $G$ with terminal set $T'$ (\cref{sec:sparsification}).
\end{enumerate}

\begin{algorithm}
    \SetKwInOut{Input}{Input}
    \SetKwInOut{Output}{Output}
    
    \Input{Undirected weighted graph $G$, terminal set $T \subseteq V$, $\gamma=1/\textup{polylog}(n)$, $k = \textup{polylog}(n)$}

    \For{$i \gets 1$ \KwTo $O(\log n)$}{
        $U \gets T$
    
        $\Tilde{\lambda}\gets 2^i$

        \Do{$|U|>k$}{
            Run algorithm from Unbalanced Case (\cref{lemma:unbalanced_case})
            with terminal set $U$ \label{line:unbalanced}\tcp{Unbalanced Case} 
            
            Run \textsc{Terminal-Decomp($V, U, \Tilde{\lambda}, \gamma$)}
            
            Find sparsified set $U' \subset U$ \tcp{\cref{theorem:sparsification}, Balanced Case}

            $U \gets U'$
        }
    }

    \Return{minimum weight Steiner cut over all iterations of \cref{line:unbalanced}}
    \caption{\textsc{Minimum-Steiner-Cut($G, T$)}}
    \label{alg:steiner-min-cut}
\end{algorithm}

We give a high-level analysis of the runtime, which we formally prove in the following sections. Each call of terminal decomposition takes polylogarithmic max-flow computations and at most near-linear time with respect to the graph outside of the max-flows. Since the sparsification procedure halves the terminal set each iteration, it adds at most a $\log n$ extra factor in runtime. Along with the additional $\log n$ factor for guessing $\Tilde{\lambda}$, this gives us our claimed runtime.
\section{Terminal Decomposition Using Cut-Matching Game} \label{sec:cut_game}
The goal of the cut-matching game is to try to certify that the entire vertex set $V$ is $(s, \delta, \gamma, T)$-terminal-strong in $G$ by iteratively constructing our cut-graph $H$ to be $(s, \Tilde{O}(\delta), \gamma)$-strong. This may not always be possible, but throughout the cut-matching game, the algorithm may also verify that $V$ is $(s,\delta,\gamma,C)$-terminal-strong for a subset $C\subseteq T$ with $|C|\ge2|T|/3$. In that case, we apply a \emph{trimming} procedure similar to \cite{expander-pruning}. Otherwise, if this is also not possible, we will be able to find a balanced sparse cut in the cut-graph $H$. We then run a max-flow between the two terminal sets in $G$, which outputs either a large flow or (by duality) a small cut. In the former case, we add a corresponding large (fractional) matching to the cut graph $H$. In the latter case, we immediately find a balanced terminal-sparse cut in $G$, at which point we recursively decompose the two sides.

Before going through the formal procedure, we give high-level overviews of the cut and matching player strategies. 

\paragraph*{Cut Player Description}
The cut player attempts to find a sparse, balanced cut $(U,T\setminus U)$ in the cut graph, which ensures that we make sufficient progress when the matching player creates a matching. If the cut player fails to find a cut, we terminate the cut-matching game and prove that the original graph $G$ satisfies desirable properties.

\begin{figure*}[ht]
    \centering
    \fbox{
    \begin{minipage}{0.9\textwidth}
    \textbf{Cut Player Strategy on current cut-graph $H$}

    \begin{itemize}
        \item Find an $(s, \delta, \gamma)$-strong decomposition on cut graph $H$
        \item If a cluster with size greater than $2|T|/3$ exists, we terminate the cut-matching game. We trim the cluster according to \Cref{alg:trimming} and certify the cluster $U$ as $(s, \delta, \gamma)$-terminal-strong. We then recursively apply \cref{alg:terminal_decomposition} on the smaller side. 
        
        \item Otherwise, we merge the clusters into two groups that each contains between 1/3 and 2/3 fraction of all vertices of $H$ (this is always possible, see \cref{claim:large_clusters}). Denote the bipartition as $(C, \overline C)$.
    \end{itemize}
   \end{minipage}
    }
\end{figure*}

\paragraph*{Matching Player Description}
The matching player's goal is to add edges in the cut graph corresponding to the maximum possible flow in $G$ from one side of the bipartition to the other. They do this by running a maximum flow algorithm across the bipartition. The matching player adds edges to the cut graph if a large flow is successfully routed. Otherwise, a terminal-balanced cut is found, and we terminate the cut-matching game and recursively apply our terminal-strong-decomposition algorithm on both sides of the cut.

\begin{figure*}
    \centering
    \fbox{
    \begin{minipage}{0.9\textwidth}
    \textbf{Matching Player Strategy on bipartition $C$ of cut-graph $H$}
    
        \begin{itemize}
           \item We calculate a max-flow on graph $G$ between the terminals in $C$ and $T \setminus C$ using \cref{alg:trimming}. If the flow has value at least $|T|/6 \cdot \delta \cdot \psi$, we call the flow a ``large flow''. Otherwise, the flow has value less than $|T|/6 \cdot \delta \cdot \psi$, so we call the corresponding cut a ``small cut''. 
           \item If we find a large flow, we add a large matching into cut graph $H$: we break down the flow into paths and add edges between vertices in graph $H$ with the same corresponding weights as the flow paths.
           \item If we find a small cut, we certify the minimum cut found as a terminal-balanced, terminal-sparse cut. We stop the cut-matching game and recursively apply our terminal-decomposition algorithm on both sides.
        \end{itemize}
    \end{minipage}
    }
\end{figure*}

We formally define strategies for the cut and matching players in this game in \cref{alg:cut-game}. Our guarantee given by our cut-matching game method is stated as the following:
\begin{lemma}\label{lemma:cut_game_guarantee}
    Given an undirected weighted graph $G$ and parameters $\delta>0$, $\psi=1/\textup{polylog}(n)$, \cref{alg:cut-game} runs in time $\tilde O(m)$ plus $\textup{polylog}(n)$ calls to maximum flow, and outputs one of the following:
    \begin{enumerate}
        \item An $(O(\log^9n/\psi^5), \delta, \Omega(\psi^5/\log^9n),T)$-terminal-strong cluster $U$ with $|U\cap T|\ge|T|/3$ such that $\overline U$ is either empty or $\psi\cdot\delta$-terminal-sparse, or
        \item A $\psi\cdot\delta$-terminal-sparse cut $(U, \overline{U})$ with $|U\cap T|,|\overline U\cap T|\ge|T|/6$
    \end{enumerate}
\end{lemma}

The correctness of the Cut Player strategy is shown in \cref{subsec:cut_player}, matching player strategy in \cref{subsec:matching_player}, and the termination within $L_{\text{max}}$ rounds is proved in \cref{subsec:termination}.

\begin{algorithm}
    \SetKwInOut{Input}{Input}
    \SetKwInOut{Output}{Output}
    \Input{Undirected weighted graph $G$, terminal set $T \subseteq V$, decomposition parameters $\delta>0$ and $\psi<1$.}
    \Output{Either \textbf{(1)} a $(\psi\cdot\delta)$-terminal-sparse cut $(U, \overline{U})$ with $|U\cap T|,|\overline U\cap T|\ge|T|/6$, or\linebreak \textbf{(2)} an $(\Tilde{O}(\psi^{-5}), \delta, \Tilde{\Omega}(\psi^5),T)$-terminal-strong cluster $U$ with $|U\cap T|\ge|T|/3$ such that $\overline U$ is either empty or $(\psi\cdot\delta)$-terminal-sparse}
    
    Initialize cut graph: $H=(T,\emptyset)$

    Initialize unmatched terminal set: $S \gets \emptyset$

    Set parameters $L_{\max}=O(\log|T|)$, $\alpha=L_{\max}/\psi$, $s=O((L_{\max}/\psi)^2\log^2n)$, $\gamma=\frac1{200\alpha s}$
    
    \For{$t \gets 1$ \KwTo $L_{\max}$} {
        \tcp{Cut Player Strategy}
        Partition $T$ into clusters that are $(s, \alpha\delta, \gamma)$-strong in $H$. (\Cref{lemma:partition}) \label{line:partition}

        \uIf{there exists cluster $C$ of size $\geq 2|T|/3$}{
            \tcp{Trim Cluster}
            $(f, (U, \overline U)) \gets \textsc{CutOrFlow}(G, C, \min\{\gamma/2s,\gamma/6\})$ \label{line:trimming}
            
            \Return larger side $U$ with cut $(U,\overline U)$ if $\overline U$ is non-empty\label{line:larger-side-U}
        }

        Combine clusters to create a bipartition $(C, \overline{C})$ with $|C| \ge |\overline{C}| \ge |T|/3$. (\Cref{claim:large_clusters}) \label{line:combine}

        \tcp{Matching Player Strategy}
        $(f, (U, \overline U)) \gets \textsc{CutOrFlow}(G, C, \psi)$

        \uIf{flow $f$ has value $\ge |T|/6 \cdot\delta\cdot\psi$}{ \label{line:if-flow}
            \tcp{Large Flow}
            Decompose flow into (implicit) paths $f_1, f_2, ..., f_k$ where $k\leq m$ and each path connects exactly two terminals ($t_1, t_2$), one from each bipartition
        
            For each $i\in[k]$, add an edge $(t_1, t_2)$ in $H$ whose weight is $1/\psi$ times the capacity of path $f_i$ in $G$ \label{line:add-edges}
        }
        \Else{
            \tcp{Terminal-Balanced Cut}
            \Return terminal-balanced cut $(U,\overline U)$
        }
    }
    \tcp{Game must terminate within $L_{\max}$ rounds}
    \caption{\textsc{Cut-Game($G, T, \delta, \psi$)}}
    \label{alg:cut-game}
\end{algorithm}

\subsection{Cut Player}
\label{subsec:cut_player}
The lemma below for $\alpha=L_{\max}/\psi$ shows that \Cref{line:partition} of \cref{alg:cut-game} can be computed efficiently. We apply the lemma on graph $H$ and vertex set $T$.
\begin{lemma}\label{lemma:partition}
Given any parameters $\delta>0$ and $\alpha\le\textup{polylog}(n)$ and a graph $G=(V,E)$ with total edge weight at most $\alpha\delta n$, there exists $s\le O(\alpha^2\log^2n)$ and $\gamma=\Omega(1/s)$ and an algorithm in $\tilde O(m)$ time that outputs a decomposition of $V$ into $(s, \alpha\delta, \gamma)$-strong clusters such that the total weight of inter-cluster edges is at most $n\delta/50$.
\end{lemma}

The first step is to apply the following lemma to a slightly modified graph.

\begin{lemma}[\cite{near_linear_mincut}]\label{lem:strong-partition}
Given a weighted graph $G=(V,E,w)$ and a parameter $\delta_0$ such that $\delta_0\leq \min_{v\in V}\deg(v)$ and a parameter $s_0\le \delta_0\textup{polylog}(n)$, there is an algorithm that runs in $\tilde O(m)$ time and partitions the vertex set $V$ into components $V_1,\ldots,V_k$ such that
 \begin{enumerate}
 \item For any cluster $V_i$ and any cut $(S,\overline S)$ in $G$ of weight at most $\delta_0$, we have $\min\{\textbf{\textup{vol}}(S\cap V_i),\textbf{\textup{vol}}(\overline S\cap V_i)\}\le s_0$. Here, $\textbf{\textup{vol}}(U)$ is the sum of weighted degrees of vertices in $U$.
 \item The total weight of inter-cluster edges is at most an $O(\frac{\sqrt{\delta_0} \log n}{\sqrt {s_0}})$ fraction of the total weight of edges.
 \end{enumerate}
\end{lemma}

Construct the graph $G_0$ from $G$ as follows. For each vertex $v\in V$, add a new vertex $v'$ with an edge to $v$ of weight $\alpha\delta$. This new graph has minimum weighted degree $\alpha\delta$. Apply the lemma above to $G_0$ with parameters $\delta_0=\alpha\delta$ and $s_0=s\alpha\delta$. The total weight of inter-cluster edges is at most $O(\frac{\log n}{\sqrt s})\cdot \alpha\delta n \le n\delta/100$ for large enough $s=O(\alpha^2\log^2n)$. Since $G_0$ has minimum degree $\delta_0$, the guarantee $\min\{\textbf{\textup{vol}}(S\cap V_i),\textbf{\textup{vol}}(\overline S\cap V_i)\}\le s_0$ from property~1 implies that $\min\{|S\cap V_i|,|\overline S\cap V_i|\}\le s_0/\delta_0=s$. In other words, each $V_i$ is $(s,\alpha\delta,0)$-strong in $G_0$. Consider the partition in $G$ obtained by removing all new vertices $v'$. It is straightforward to see that this partition is also $(s,\alpha\delta,0)$-strong in $G$, and the total weight of inter-cluster edges is still at most $n\delta/100$.

We now modify the partition so that each cluster is $(s,\alpha\delta, \gamma)$-strong by applying the lemma below to each $V_i$. The total weight of additional inter-cluster edges guaranteed by the lemma is at most $\sum_i|V_i|\delta/100\le n\delta/100$. Together with the inter-cluster edges from the first step, the total weight is at most $n\delta/12$. It remains to prove the lemma below:

\begin{lemma}\label{lemma:partition-strong}
Let $C$ be an $(s,\alpha \delta,0)$-strong cluster in $G$ and let $\gamma=\frac1{200\alpha s}$. There is an algorithm in $\tilde{O}(|E(G[C])|)$ time that partitions $C$ into $(s,\alpha \delta,\gamma)$-strong clusters such that the total weight of inter-cluster edges is at most $|C|\delta/100$.
\end{lemma}
This proof uses a similar technique found in \cite{near_linear_mincut}, and we leave the proof for \cref{appendix:partition-proof}.

Let the \emph{size} of a cluster be the number of vertices in the cluster. The following lemma shows that \Cref{line:combine} of \cref{alg:cut-game} can be executed efficiently.
\begin{claim} \label{claim:large_clusters}
    Suppose no single cluster has size greater than $2|T|/3$. Then, there exists a bipartition of clusters such that each group of clusters has a total size in the range $[|T|/3, 2|T|/3]$, and this bipartition can be computed in nearly linear time.
\end{claim}
\begin{proof}
    We can split our proof into two cases:
    \begin{enumerate}
        \item There exists a cluster with size $\in [|T|/3, 2|T|/3]$:

        We make that cluster its own group and all remaining clusters the second group.

        \item All clusters have size less than $|T|/3$:

        Enumerate the clusters in an arbitrary order, and consider the shortest prefix of clusters whose total size exceeds $|T|/3$. The prefix without its last cluster has total size less than $|T|/3$, and this last cluster of the prefix has size less than $|T|/3$, so this prefix has size in $[|T|/3, 2|T|/3]$. \qedhere
    \end{enumerate}
\end{proof}

\subsection{Matching Player}
\label{subsec:matching_player}
We introduce a key subroutine of the matching player, called \textsc{CutOrFlow}, which is used to find matchings over partitions and for trimming.
\begin{algorithm}
    \SetKwInOut{Input}{Input}
    \Input{Undirected weighted graph $G=(V,E)$ and a set of terminals $S \subseteq T$}

    $G' \gets G$

    Add sink node $t$ to $G'$ with edges from all terminals in $T\setminus S$ with capacity $\delta\cdot \kappa$. Add source node $s$ to $G'$ with edges to all terminals in $S$ with capacity $\delta\cdot \kappa$. 
    
    Compute $s-t$ max-flow $f'$ and $s-t$ min-cut $(U', \overline{U'})$ in $G'$

    Let $f$ be the flow $f'$ with vertices $s, t$ removed from every path

    Let cut $(U, \overline{U})$ in $G$ be $(U' \setminus s, \overline{U'} \setminus t)$
    
    \Return flow $f$ and cut $(U, \overline{U})$

    \caption{\textsc{CutOrFlow($G, S, \kappa$)}}
    \label{alg:trimming}
\end{algorithm}

\begin{lemma}\label{lemma:sparsity}
    If \cref{alg:trimming} returns a valid cut $(U, \overline{U})$, then it is $\delta\cdot \kappa$-terminal sparse in $G$.
\end{lemma}
\begin{proof}
    Without loss of generality, assume that $U$ contains at most as many terminals as $V \setminus U$. Consider the difference between the edges of $\partial U'$ and the edges of the cut $\partial(\{s\})$, which cuts all edges adjacent to $s$. The edges in $\partial U' \setminus \partial(\{s, t\})$ are precisely the edges of $\partial U'$ originally within $G$, and the edges in $\partial(\{s\}) \setminus \partial U'$ are precisely the edges between $s$ and $U \cap T$. Since $\partial U'$ is an $s-t$ min-cut, we have $w(\partial U' \setminus \partial(\{s, t\})) \le w(\partial(\{s\}) \setminus \partial U')$, which is equivalent to $w_G(\partial U) \le |U \cap T|\cdot\delta\cdot \kappa$. A symmetric argument yields $w_G(\partial U)\le|\overline U\cap T|\cdot\delta\cdot\kappa$, and combining the two proves the lemma.
\end{proof}

\begin{figure}
\includegraphics[width=9cm]{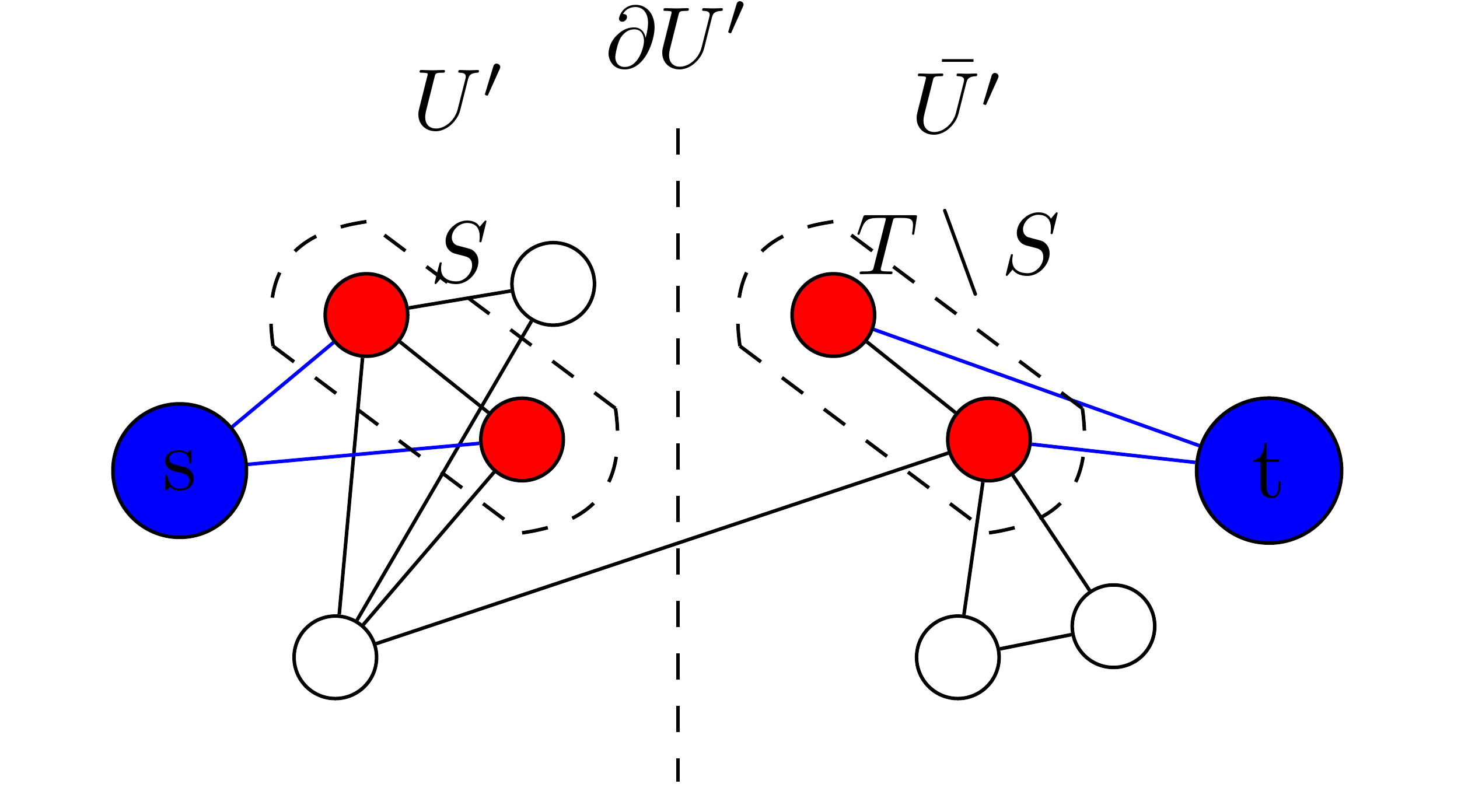}
\centering
\caption{Construction of \cref{alg:trimming}. Blue vertices and edges are added to the original graph $G$ and red vertices mark terminals.}
\end{figure}

\cref{lemma:sparsity} proves the sparsity guarantees of \cref{lemma:cut_game_guarantee}, as the algorithm sets either sets $\kappa \gets \psi$ or $\kappa \gets \min\{\gamma/2s,\gamma/6\}) \ll \psi$ in every \textsc{CutOrFlow} call. Thus, every cut returned is always $\psi \cdot \delta$-terminal sparse in $G$.

\begin{lemma}\label{lemma:constant_size}
If $|S|,|T\setminus S|\ge|T|/3$ and flow $f$ has value less than $|T|/6 \cdot \delta \cdot \kappa$, then the cut $(U, \overline U)$ satisfies $|U|, |\overline U| \ge |T|/6$.
\end{lemma}
\begin{proof}
In graph $G'$, the value of flow $f'$ and cut $(U', \overline{U'})$ are equal by flow-cut duality. In particular, cut $(U', \overline{U'})$ has weight less than $|T|/6 \cdot \delta \cdot \kappa$. Since $s$ has edges to $S\cap\overline{U'}$ that cross the cut, and since $t$ has edges from $(T\setminus S)\cap U'$ that cross the cut, we have $|S\cap\overline{U'}|,|(T\setminus S)\cap U'|\le|T|/6$. Since $|S|,|T\setminus S|\ge|T|/3$, it follows that $|S\cap U'|=|S|-|S\cap\overline{U'}|\ge|T|/6$ and $|(T\setminus S)\cap\overline{U'}|=|T\setminus S|-|(T\setminus S)\cap U'|\ge|T|/6$. In particular, $|U|, |\overline U| \ge |T|/6$.
\end{proof}

\subsubsection{Trimming}\label{subsubsec:trimming}
In \Cref{line:trimming} of \cref{alg:cut-game}, we begin with a subset $C\subseteq T$ of size at least $2|T|/3$ such that $V$ is $(s,\delta,\gamma,C)$-terminal-strong in $G$. Our next goal is to find a cluster $U$ that is a $(O(s/\gamma),\delta,\Omega(\gamma/s))$-terminal strong with $|U\cap T|\ge|T|/3$. This allows us to only recurse on $\overline U$, which satisfies $|\overline U\cap T|\le2|T|/3$, allowing for an efficient algorithm.

We used a modified form of the trimming method found in \cite{expander-pruning}. In their paper, the authors describe a simple ``Slow Trimming'' and an improved ``Efficient Trimming'' scheme, which is much more involved by circumventing the use of exact max-flow. However, the slow trimming scheme suffices for our purposes since we are fine with maximum flow time.

We begin with the following lemma, which we use to prove the correctness of trimming:

\begin{lemma} \label{lem:condition1_induced_expander}
    If a cluster $S$ is $(s, (L_{\max}/\psi)\delta, \gamma)$-strong in the cut graph $H$, then $V$ is $(s, \delta, \gamma, S)$-terminal-strong in $G$.
\end{lemma}
\begin{proof}
    By construction, each edge $(u,v)$ of weight $w$ in the cut graph $H$ certifies the existence of a flow of capacity $1/\psi \cdot w$ in the original graph $G$. Since we run the cut-game for at most $L_{\max}$ rounds, we can simultaneously route flows between terminals $u$ and $v$ of weight $1/\psi \cdot w(u,v)$ for all edges $(u,v)\in E_H$ with capacities scaled by at most $L_{\max}$ in graph $G$. Equivalently, scaling everything by $1/\psi$, we can simultaneously route flows between terminals $u$ and $v$ of weight $w(u,v)$ for all edges $(u,v)\in E_H$ with capacities scaled by at most $L_{\max}\cdot1/\psi$ in graph $G$.
    
    We proceed with two cases. First, assume for contradiction that there exists a Steiner cut $(C, \overline{C})$ of graph $G$ with at most weight $\delta$ which satisfies $\min\{|C\cap S|, |\overline{C} \cap S|\} > s$. Consider cut $(C\cap T, \overline C\cap T)$ in cut graph $H$. Since $C\cap S \subseteq C \cap T$ and $\overline{C} \cap S \subseteq \overline C\cap T$, we have
    \begin{equation}
        \min\{|(C\cap T) \cap S|, |(\overline C\cap T) \cap S|\} \ge \min\{|C\cap S|, |\overline{C} \cap S|\} > s.
    \end{equation}
    The weight of cut $(C\cap T, \overline C\cap T)$ in $H$ is at most a $L_{\max}\cdot1/\psi$ factor greater than the amount of (scaled) flow able to be routed over cut $(C, \overline{C})$ in graph $G$. In other words, $w_H(C\cap T, \overline C\cap T) \leq (L_{\text{max}}\cdot1/\psi)\delta$. This contradicts the assumption that $S$ is an $(s, (L_{\max}/\psi)\delta, \gamma)$-strong cluster in the cut graph.

    For the second case, assume for contradiction that there exists a cut $(C, \overline{C})$ of graph $G$ with at most weight $\delta$ which satisfies $w(\partial_G C) < \gamma \cdot \delta$ and $\min\{|C\cap S|, |\overline{C} \cap S|\} > 0$. Similar to above, the weight of cut $(C\cap T, \overline C\cap T)$ in $H$ is at most a $L_{\text{max}}\cdot1/\psi$ factor larger than cut $(C, \overline{C})$ in graph $G$. Since $\min\{|C\cap S|, |\overline{C} \cap S|\} > 0$, $\partial_{H[S]} C$ is an actual cut of $H[S]$, and $\partial_{H[S]} C \leq \partial_H C < \gamma \cdot (L_{\max}/\psi)\delta$, contradicting the assumption that $S$ is $(s, (L_{\max}/\psi)\delta, \gamma)$-strong in $H$.
\end{proof}

Now, we introduce the section's main theorem, which shows that the cluster $\overline U$ is terminal-strong and contains a large fraction of terminals. 

\begin{theorem}\label{theorem:trimming_main}
    If $V$ is $(s, \delta, \gamma, S)$-terminal-strong in $G$ and $|S| \geq 2|T|/3$, the cut $(U, \overline U)$ returned by \cref{alg:trimming} with parameter $\kappa=\min\{\gamma/(2s),\gamma/6\}$ satisfies the property that $U$ is $(\max\{2/\kappa+s,3s\}, \delta, \kappa, U\cap T)$-terminal strong in $G$ and $|U\cap T|\geq|T|/3$.
\end{theorem}
We split our proof into two cases, when the cut $U = V$, and all other cuts.

\paragraph*{Case 1: $U = V$.}
Here, we will only use the bound $\kappa\le\gamma$. This fact will be important later in the proof.

By flow-cut duality, the flow $f$ in $G'$ sends full capacity along each edge into $t$. In particular, the value of the flow is equal to $|T \setminus S| \cdot\delta\cdot\kappa$. This case clearly satisfies $|U \cap T| = |T| \geq |T|/3$. 

We now show that in this case $U$ is $(\max\{2/\kappa+s,3s\}, \delta, \kappa, U\cap T)$-terminal strong. Consider an arbitrary Steiner cut $(C, \overline C)$ in $G$ of size $<\delta$. Suppose first that $\min\{|C\cap S|, |\overline{C} \cap S|\} = 0$, and assume without loss of generality that $C \cap S = \emptyset$. Each terminal in $C \cap (T \setminus S)$ sends full capacity into $t$ in the flow $f$, so at least $|C \cap (T \setminus S)|\cdot\delta\cdot \kappa$ flow must cross the cut $C$. Since $w(\partial C) < \delta$, we obtain $|C \cap (T \setminus S)|\cdot\delta\cdot \kappa<\delta$, so $|C \cap T| = |C \cap (T \setminus S)| < 1/\kappa$. Additionally since $|C \cap (T \setminus S)| \ge 1$, we have the total flow being at least $\kappa \cdot\delta$, and therefore the cut is at least this size as well.

Suppose now that $\min\{|C\cap S|, |\overline{C} \cap S|\} > 0$. From \cref{lem:condition1_induced_expander}, we know that $\min\{|C\cap S|, |\overline{C} \cap S|\} \le s$ and $w(\partial C) \geq \gamma \cdot \delta \geq \kappa\cdot\delta$. Assume without loss of generality that $|C\cap S| \leq |\overline{C} \cap S|$. We consider two cases:
\begin{enumerate}
    \item Case 1a: $|C\cap(T\setminus S)| \geq 2|C\cap S|$.

    Recall that the $s-t$ flow $f$ has value $|T\setminus S|\cdot\delta\cdot \kappa$ in graph $G'$. In this flow, at most $|C\cap S|\cdot\delta\cdot \kappa$ of the flow initially routed into $C\cap(T\setminus S)$ from $s$ can reach $t$ without crossing cut $C$. The remaining flow must therefore cross cut $C$. Therefore we have 
    \begin{equation}
         \delta > w(\partial C) \geq (|C\cap(T\setminus S)|-|C\cap S|)\cdot\delta\cdot \kappa \geq |C\cap(T\setminus S)|/2\cdot\delta\cdot\kappa.
    \end{equation}

    Therefore $|C\cap(T\setminus S)| < 2/\kappa$, and thus $|C \cap T| = |C\cap(T\setminus S)| + |C\cap S| < 2/\kappa + s$.

    \item Case 1b: $|C\cap(T\setminus S)| < 2|C\cap S|$.
    
    We have $|C\cap(T\setminus S)| < 2|C\cap S| \leq 2s$, so $|C \cap T| = |C\cap(T\setminus S)| + |C\cap S| \leq 3s$.
\end{enumerate}

This completes the proof of case 1 of \cref{theorem:trimming_main} when $U = V$.

\paragraph*{Case 2: $U \subsetneq V$.}
We begin by showing $|U\cap T|\geq|T|/3$. If this was not the case, since we assume $|S|\ge2|T|/3$, more than $|T|/3$ terminals in $S$ would be in $\overline U$. All of these terminals would have an edge of size $\delta\cdot\kappa$ crossing the cut $\partial_{G'} U$. However, the $s-t$ cut $\partial_{G'} U$ must have size at most $|T\setminus S|\cdot\delta\cdot\kappa$ in graph $G'$. Since $|T\setminus S| \leq |T|/3$, we arrive at a contradiction. 

Now we prove the terminal-strong property of $G[U]$. We start by proving the following lemma:
\begin{lemma}\label{lem:condition2_induced_expander}
    Assume that $V$ is $(s, \delta, \gamma, S)$-terminal-strong in $G$. Then $U$ is $(s, \delta, \gamma/6, S\cap U)$-terminal-strong in $G$.
\end{lemma}
\begin{proof}
    Note that the condition $\min\{|C \cap U \cap S|, |\overline C \cap U \cap S|\} \le s$ follows immediately from $\min\{|C \cap S|, |\overline C \cap S|\} \le s$ since $G$ is $(s, \delta, \gamma, S)$-terminal-strong. So it suffices to prove that for all Steiner cuts $(C, \overline C)$ such that $\partial C \leq \delta$, we have $w(E(C \cap U, \overline C \cap U))\geq \gamma/6 \cdot\delta$ if $\min\{|C \cap U \cap S|, |\overline C \cap U \cap S|\} > 0$. 

    By flow-cut duality, the flow $f$ saturates the entire boundary $E(U, \overline U)$. Assume for contradiction there exists a Steiner cut $(C, \overline C)$ such that $\partial C \leq \delta$, $w(E(C \cap U, \overline C \cap U)) < \gamma/6\cdot\delta$, and $\min\{|C \cap U \cap S|, |\overline C \cap U \cap S|\} > 0$. As mentioned before, we know that $\min\{|C \cap U \cap S|, |\overline C \cap U \cap S|\} \le s$, so assume without loss of generality that $|C\cap U\cap S|\le s$. We also have $w(\partial(C\cap U)) = w(E(C \cap U, \overline C \cap U)) + w(E(C \cap U, V\setminus U)) \geq \gamma\cdot\delta$ since $G$ is $(s, \delta, \gamma, S)$-terminal-strong. This implies $w(E(C\cap U, V\setminus U)) > 5\gamma/6\cdot\delta$. The flow $f$ sends $>5\gamma/6\cdot\delta$ flow from $C\cap U$ to $V \setminus U$, and only $<\gamma/6\cdot\delta$ flow can enter $C\cap U$ from $\overline C \cap U$. Therefore $>2\gamma/3\cdot\delta$ flow must be routed from the source node $s$ into $C\cap U$. But this flow is upper bounded by $|C \cap U \cap S| \cdot \delta \cdot \kappa \leq s\cdot\delta\cdot\gamma/(2s)\le \gamma/2\cdot\delta < 2\gamma/3\cdot\delta$ (using our assumption $\kappa\le\gamma/(2s)$ from \Cref{theorem:trimming_main}), and we arrive at our contradiction.
\end{proof}

Next, we show that $|S \cap U| \geq 2|T \cap U|/3$. For each terminal from $S$ in $V \setminus U$ and each terminal from $T \setminus S$ in $U$, there exists an edge of weight $\delta\cdot 1/\kappa$ crossing $\partial U$ in $G'$. Denote $|S \cap (V\setminus U)|=a$ and $|(T \setminus S)\cap U|=b$. Since the cut $(U,\overline U)$ has weight $<|T\setminus S|\cdot\delta\cdot\kappa\leq |T|/3\cdot\delta\cdot\kappa$, we have $a+b<|T|/3$. Additionally $|S| \geq 2|T|/3$, so 
\begin{equation}
    \frac{|(T\setminus S)\cap U|}{|S\cap U|} = \frac b{|S|-a} \leq \frac{a+b}{|S|} < \frac{|T|/3}{2|T|/3} \leq 1/2,
\end{equation}
proving the requirement as desired.

At this point, we can apply the $U=V$ case on $G[U]$, since we have shown that $G[U]$ is $(s, \delta, \gamma/6, S\cap U)$-terminal-strong and $|S \cap U| \geq 2|T \cap U|/3$, and the $U=V$ case only requires that $\kappa\le\gamma/6$. This completes the proof of case 2 of \cref{theorem:trimming_main} when $U \subsetneq V$.

\subsubsection{Final Parameters}
Finally, we plug in our parameters $L_{\max}=O(\log|T|)$, $\alpha=L_{\max}/\psi$, $s=O((L_{\max}/\psi)^2\log^2n)$ $=O(\log^4n/\psi^2)$, and $\gamma=\frac1{200\alpha s}=\Omega(\psi^3/\log^5n)$ in \Cref{alg:cut-game}. We have $\kappa=\Omega(\gamma/s)=\Omega(\psi^5/\log^9n)$, so the $(\max\{2/\kappa+s,3s\}, \delta, \kappa, U\cap T)$-terminal strong cluster $U$ output by \Cref{alg:cut-game} is \\$(O(\log^9n/\psi^5),\delta,\Omega(\psi^5/\log^9n))$-terminal-strong, fulfilling the output guarantee of \Cref{alg:cut-game}.

\subsection{Termination}
\label{subsec:termination}
To show that the cut-matching game terminates within $L_{\max}$ rounds, we introduce the following guarantee of the cut-matching game analysis.
\begin{lemma}\label{lemma:cut_graph_expander}
    For large enough $L_{\max}=O(\log|T|)$, \cref{alg:cut-game} proceeds for at most $L_{\max}$ iterations.
\end{lemma}
The proof is a direct adaptation of the cut-matching game analysis of \cite{cut-matching-game, single-commodity-flows}, and thus we leave the details to \cref{appendix:cut-game}.

\subsection{Terminal Decomposition}
Finally, we introduce the complete algorithm for terminal decomposition, which uses the cut-matching game algorithm as a key subroutine.

\begin{algorithm}
    \SetKwInOut{Input}{Input}
    \SetKwInOut{Output}{Output}
    
    \Input{Cluster $C\subseteq V$, terminal set $T\subseteq V$, decomposition parameters $\delta, \psi$}
    
    Run \textsc{Cut-Game($G[C], T, \delta, \psi$)}
            
    \uIf{$G[C]$ is certified as a terminal-strong cluster}{
        \Return $G[C]$
    }
    \uElseIf{Cut-Game returns balanced cut $(U, \overline{U})$ with terminal-sparsity $\psi\cdot\delta$}{
        \Return $\textsc{Terminal-Decomp}(U, U \cap T, \delta, \psi) \cup \textsc{Terminal-Decomp}(\overline{U}, \overline{U} \cap T, \delta, \psi)$ \label{line:recurse}
    }
    \Else{
    \tcp{Larger side $G[U]$ is certified as an $(s, \delta, \gamma)$-terminal-strong cluster}
        \Return $G[U] \cup \textsc{Terminal-Decomp}(\overline{U}, \overline{U} \cap T, \delta, \psi)$ 
    }
    \caption{\textsc{Terminal-Decomp($C, T, \delta, \psi$)}}
    \label{alg:terminal_decomposition}
\end{algorithm}

The algorithm uses \textsc{Cut-Game} as a subroutine and \cref{lemma:cut_game_guarantee} as its guarantee. First, we note that since we recurse on both sides of a cut in \cref{line:recurse} only if they both have at least $\Omega(1)$ terminals (from \cref{lemma:cut_game_guarantee}), we have at most $O(\log n)$ recursive levels. We now prove the formal theorems for our terminal decomposition.

\begin{theorem}
    \cref{alg:terminal_decomposition} runs with $O(\log^2 n)$ max-flows and $\Tilde O(m)$ additional time.
\end{theorem}
\begin{proof}
    First, in each recursive level, each \textsc{Terminal-Decomp} call is on a mutually disjoint portion of the graph. Therefore, all maximum flows on a single recursive level can be done in parallel with a single maximum flow call on a graph of size $O(m)$ edges. Additionally, each round of the cut-matching game uses at most a single max-flow call (in \textsc{CutOrFlow}). With a total of $L_{\text{max}} = O(\log n)$ cut-matching game rounds and $O(\log n)$ recursive levels, the entire terminal decomposition runs in $O(\log^2 n)$ max-flows.

    All other cut-matching game procedures (specifically the $(s,\delta,\gamma)$-strong decomposition of \linebreak\cref{lemma:partition}) run in near-linear time. With $O(\log n)$ recursive levels, the entire algorithm runs in near-linear time, excluding max-flows.
\end{proof}
\begin{theorem}\label{thm:terminal-decomposition}
    \cref{alg:terminal_decomposition} returns a $(\Tilde{O}(1/\psi^5), \delta, \Tilde{\Omega}(\psi^5), T)$-terminal-strong decomposition of $G$.
\end{theorem}
\begin{proof}
    We begin by proving the upper-bound on intercluster edges:
    \begin{lemma}
        The total weight of intercluster edges from the decomposition outputted by \cref{alg:terminal_decomposition} is at most $O(\psi\cdot\delta\cdot|T|\log |T|)$. 
    \end{lemma}
    \begin{proof}
        Every cut made by \cref{alg:terminal_decomposition} is $\psi\cdot\delta$ terminal-sparse due to \cref{lemma:sparsity}. We can charge $\psi\cdot\delta$ weight to each terminal on the smaller side of the cut. Since there are at most $\log |T|$ recursive levels and each cluster gets only one cut per recursive level, each terminal gets charged at most $\log |T|$ times. Summing up the weights charged to each terminal gives us a total edge weight of $O(\psi\cdot\delta\cdot|T|\log |T|)$.
    \end{proof}

    From \cref{lemma:cut_game_guarantee}, every cluster returned is certified to be $(O(\log^9n/\psi^5), \delta, \Omega(\psi^5/\log^9n), T)$-terminal strong, completing the proof.
\end{proof}

\section{Minimum Steiner Cut Using Sparsification} \label{sec:sparsification}

We complete our algorithm by showing a polylogarithmic maximum flow algorithm for minimum Steiner cut on a graph by using its terminal-strong decomposition. We use the minimum isolating cuts method and terminology described by \cite{steiner-polylog-flows}. The critical difference is that we use a terminal-strong decomposition instead of an expander decomposition. However, we prove that the same guarantees apply. 

We begin by introducing some definitions from \cite{steiner-polylog-flows}.

\begin{definition}[$k$-unbalanced, $k$-balanced]
    A subset of vertices $U \subseteq V$ is considered $k$-unbalanced if there exists a minimum Steiner cut $S$ such that $\min\{|S \cap U|, |\Bar{S}\cap U|\} \leq k$. $U$ is considered $k$-balanced with witness $(S, \Bar{S})$ if there exists a minimum Steiner cut $S$ such that $\min\{|S \cap U|, |\Bar{S}\cap U|\} > k$.
\end{definition}

Our main result of the section is as follows:

\begin{theorem}\label{theorem:sparsification}
    There exists a deterministic algorithm which given an undirected weighted graph $G=(V,E)$, an $(s, \delta, \gamma, T)$-terminal-strong decomposition $G'=\{V_1, V_2,...,V_\ell\}$, a parameter $k\gets C\log^Cn$ for some large enough constant $C>0$, and a subset of terminals $U \subseteq T$, does the following:
    \begin{enumerate}
        \item If $U$ is $k$-unbalanced, we return the minimum Steiner cut of $G$ with polylogarithmic maximum flow calls and near-linear additional runtime.
        \item If $U$ is $k$-balanced with witness $(S_1, S_2)$, we return a subset $U' \subset U$ such that $|U'|\leq|U|/2$ and $S_i\cap U' \neq \emptyset$ for both $i=1,2$.
    \end{enumerate}
\end{theorem}
We leave the full proof details to \cref{appendix:sparsification}.

After computing the sparsified set $U\gets U'$ in the balanced case, we can recursively run our minimum Steiner cut algorithm on graph $G$ and terminal set $T \gets U$. Since the size of $U$ at least halves each time we sparsify, this only needs to be done at most $\log n$ times.

Also, since we never know which case we are specifically in (balanced or unbalanced) in \cref{theorem:sparsification}, we run both cases until $U$ must be guaranteed to be $k$-unbalanced. At this point we take the minimum over all Steiner cuts found, and we are guaranteed to have found a minimum one (see \cref{alg:steiner-min-cut}). 

\section{Conclusion}
Our algorithm solves deterministic minimum Steiner cut with polylogarithmic max flow calls and near-linear additional processing time. We thus show minimum Steiner cut reduces to maximum flow up to polylogarithmic factors in runtime. Specifically, the existence of a deterministic near-linear time $s-t$ max-flow algorithm would imply a deterministic near-linear time algorithm for minimum Steiner cut.

Our main contribution is the $(s, \delta, \gamma)$-terminal-strong decomposition. We are able to do this deterministically in polylogarithmic max flows and near-linear additional time for small $\delta$, which is not yet known for standard expander decompositions. We also believe that $(s, \delta, \gamma)$-strong and terminal-strong decompositions may have additional future applications in faster algorithms for graph problems.

\section*{Acknowledgements}
We want to thank Monika Henzinger, Satish Rao, and Di Wang for helpful discussions related
to \cref{lemma:partition-strong}.

\bibliographystyle{plain}
\bibliography{citation}
\newpage
\appendix
\section{Strong Partition Proof}\label{appendix:partition-proof}
We provide the complete proof for \cref{lemma:partition-strong} here. Within the context of this proof, we assign a separate identity to each edge, and we do not merge distinct edges upon contraction.

The algorithm begins with $H\gets G[C]$ and iteratively executes the following two steps in arbitrary order whenever possible.
 \begin{enumerate}
 \item Contract two vertices with at least $\gamma \alpha \delta$ total weight of edges between them.\label{step:partition-strong-1}
 \item Remove a vertex $v$ with weighted degree at most $\delta/100$ in $H$.\label{step:partition-strong-2}
 \end{enumerate}
At the end of the proof, we show how to perform these steps in near-linear time overall.

For each removed vertex, consider all original vertices in $C$ that were contracted to that vertex, and add a new output cluster consisting of those vertices. If $H$ is non-empty at the end of the iterative algorithm, add another cluster consisting of all vertices in $C$ that were contracted to a vertex in $H$. By construction, the resulting clusters partition $C$. By dynamically maintaining appropriate structures, the algorithm can be implemented in $\tilde{O}(|E(G[C])|)$ time.

The bound on the weight of inter-cluster edges follows from the fact that we remove at most $|C|$ many vertices in the algorithm, and each removal adds at most $\delta/100$ to the total weight of inter-cluster edges.

Since each cluster $C_i$ is a subset of $(s,\alpha \delta,0)$-strong cluster $C$, cluster $C_i$ is also $(s,\alpha \delta,0)$-strong. It remains to show that $C_i$ is $(s,\alpha \delta,\gamma)$-strong. That is, given a cut $(S,\overline S)$ in $G$ with $w(S,\overline S)\le \alpha \delta$, $S\cap C_i\ne\emptyset$, and $\overline S\cap C_i\ne\emptyset$, we have $\partial_{G[C_i]}(S\cap C_i)\ge\gamma \alpha \delta$.

First, take a set $C_i$ consisting of all vertices contracted to some vertex $v$. Color the vertices in $S\cap C_i$ black and the vertices in $\overline S\cap C_i$ white, and consider the contraction process starting from the set $C_i$ and ending at $v$, where each step contracts an edge between two vertices in the set with weight at least $\gamma \alpha \delta$. If we contract an edge whose endpoints have the same color, then assign the same color to the contracted vertex. Eventually, we contract an edge with differently colored endpoints. Each edge is included in $\partial_{G[C_i]}(S\cap C_i)$, and we contract edges of total weight at least $\gamma \alpha \delta$. It follows that $\partial_{G[C_i]}(S\cap C_i)\ge\gamma \alpha \delta$.

Now take the set $C_i$ consisting of all vertices in $C$ that were contracted to a vertex in $H$, if it is non-empty. Since $C_i$ is $(s,\alpha \delta,0)$-strong, we have $\min\{|S\cap C_i|,|\overline S\cap C_i|\}\le s$, and assume without loss of generality that $|S\cap C_i|\le s$. Color the vertices in $S\cap C_i$ black and the vertices in $\overline S\cap C_i$ white, and consider the contraction process again. If we contract an edge with differently colored endpoints, then $\partial_{G[C_i]}(S\cap C_i)\ge\gamma \alpha \delta$ as before. So suppose that never happens. At the end, let $B$ be the set of black vertices, which satisfies $\partial_HB=\partial_{G[C_i]}(S\cap C_i)$ by construction. Also, $|B|\le|S\cap C_i|\le s$ since the number of black vertices can only decrease over time. Since there are no more vertex deletions, each vertex in $B$ has weighted degree at least $\delta/100$ in $H$. Since there are no more edge contractions, the total weight of edges between black vertices is at most $\gamma \alpha \delta\binom{|B|}2$.

Suppose for contradiction that $\partial_{G[C_i]}(S\cap C_i) < \gamma \alpha \delta$, which means that
\begin{align*}
|B|\delta/100\le\textbf{\textup{vol}}_H(B)&=2w(E(H[B]))+\partial_HB
\\&=2w(E(H[B]))+\partial_{G[C_i]}(S\cap C_i)
\\&<2\cdot\gamma \alpha \delta\binom{|B|}2+\gamma\alpha \delta\le\gamma \alpha \delta|B|^2+\gamma\alpha \delta.
\end{align*}

This quadratic solves to $|B|\in\mathbb N\setminus[\ell,r]$ for some interval $[\ell,r]$. It suffices to show that $\ell\le1$ and $r\ge s$, which would imply that $|B|>s$, a contradiction. To show this claim, we simply show that the inequality fails for $|B|=1$ and $|B|=s$. For $|B|=1$, we obtain $\delta/100<2\gamma\alpha\delta$ which is false since $\gamma\le\frac1{200\alpha}$. For $|B|=s$, we obtain $s\delta/100<\gamma\alpha\delta s^2+\gamma\alpha\delta$ which is false since $\gamma\le\frac s{s^2+1}\cdot\frac1{100\alpha}$. It follows that there is no cut $(S,\overline S)$ in $G$ with $w(S,\overline S)\le \alpha \delta$, $S\cap C_i\ne\emptyset$, $\overline S\cap C_i\ne\emptyset$, and $\partial_{G[C_i]}(S\cap C_i)<\gamma \alpha \delta$.

Finally, we show that we can dynamically execute steps~(\ref{step:partition-strong-1}) and~(\ref{step:partition-strong-2}) in $\tilde O(|E(G[C])|)$ time overall. We maintain the vertex degrees, the number of edges incident to each vertex, and the total weight of edges between any two vertices, storing their values in a balanced binary tree. To execute step~(\ref{step:partition-strong-1}), query the pair of vertices with maximum total weight of edges, and to execute step~(\ref{step:partition-strong-2}), query the vertex with minimum degree. To update the maintained values over time, we perform the following. Every time a vertex is removed on step~(\ref{step:partition-strong-2}), we remove the incident edges and update values accordingly; each removed edge induces one update in each category, which is $O(|E(G[C])|)$ total updates overall. Suppose now that two vertices $u$ and $v$ are contracted, where $u$ has at most as many incident edges as $v$ (which can be checked by querying their maintained number of incident edges). We remove the contracted edges between $u$ and $v$, and for all remaining edges incident to $u$, replace the endpoint $u$ by $v$. This successfully implements step~(\ref{step:partition-strong-2}) with the contracted vertex labeled $v$. We now show that the total number of such edge updates is at most $2m\log2m$ where $m=|E(G[C])|$. For each vertex $v$, let $n_v$ be the current number of incident edges. Define the potential function
\[ \sum_{v\,:\,n_v>0}n_v\ln\frac{2m}{n_v} ,\]
which is at most $2m\log2m$ initially since $\sum_vn_v=2m$. The function $n_v\ln\frac{2m}{n_v}$ is increasing in $n_v$ in the range $n_v\in[1,m]$, which can be verified by taking the derivative:
\[ \frac d{dn_v}n_v\ln\frac{2m}{n_v}=\frac d{dn_v}n_v(\ln 2m-\ln n_v)=\ln 2m-(1+\ln n_v)=\ln\frac m{n_v}>0. \]
Since removing vertices can only decrease $n_v$, doing so can only decrease the potential. Suppose now that two vertices $u$ and $v$ are contracted with $n_u\le n_v$. After removing the contracted edges between $u$ and $v$, the values $n_u$ and $n_v$ decrease by the same amount, so $n_u\le n_v$ still. In the contraction step, we update the $n_u$ edges incident to $u$, and the $n_u\ln\frac{2m}{n_u}$ and $n_v\ln\frac{2m}{n_v}$ terms in the potential function become a single $(n_u+n_v)\ln\frac{2m}{n_u+n_v}$. The net difference is
\begin{align*}
    n_u\ln\frac{2m}{n_u}+n_v\ln\frac{2m}{n_v} - (n_u+n_v)\ln\frac{2m}{n_u+n_v} &\ge n_u\bigg(\ln\frac{2m}{n_u}-\ln\frac{2m}{n_u+n_v}\bigg)\\&\ge n_u\bigg(\ln\frac{2m}{n_u}-\ln\frac{2m}{2n_u}\bigg)=n_u ,
\end{align*}
where the second inequality follows from $n_u\le n_v$. Hence, the potential drops by at least $n_u$, while the number of edge updates is $n_u$. It follows that the total number of edge updates is at most $2m\log2m$, and each update induces one maintenance update in each category. Overall, the algorithm makes $O(m\log m)$ maintenance updates, and each update takes $O(\log m)$ time, which is $\tilde O(m)$ total as promised.

\section{Cut-Matching Proof}\label{appendix:cut-game}
For completeness, we prove \cref{lemma:cut_graph_expander} below by directly adapting the analysis of \cite{cut-matching-game}.

Let $L_{\max}=O(\log n)$ be large enough. Suppose for contradiction that the algorithm does not terminate within $L_{\max}$ iterations. On each iteration, the algorithm must execute \Cref{line:if-flow} with the partition $(C,\overline C)$ with $|C|,|\overline C|\ge|T|/3$ and $w_H(C,\overline C)<\delta |T|/12$. Following \cite{cut-matching-game}, we use the \emph{entropy function} potential
\[ \Phi_u(t)=-\sum_{v\in T}p_{u,v}(t)\log p_{u,v}(t) \qquad\text{and}\qquad \Phi(t) = \sum_{u\in T}\Phi_u(t), \]
where $p_{u,v}(t)\in[0,1]$ satisfy $\sum_{v\in T}p_{u,v}(t)=1$ for all $u\in T$. Intuitively, $p_{u,v}$ models a random walk on the cut-graph, and $p_{u,v}(t)$ is the probability distribution at time $t$ starting at vertex $u\in T$. The entropy function $\Phi_u(t)$ is at most $\ln|T|$, so $\Phi(t)\le|T|\ln|T|$ always. We will show that the potential function $\Phi(t)$ can never decrease, and it increases by $\Omega(|T|)$ every time the algorithm executes \Cref{line:if-flow}. It follows that there can only be $O(\log|T|)$ total iterations.

Let $M_t$ be the edges added to $H$ on \Cref{line:add-edges}. Since the flow $f$ sends at most $\delta\cdot\kappa$ flow through each terminal in $T$, and since the weight of the edges in $H$ are scaled by $1/\psi$, each vertex has weighted degree at most $\delta \cdot \kappa/\psi \leq \delta$ in $M_t$. Also, since $f$ has value at least $|T|/6\cdot\delta\cdot\psi$ the total weight of $M_t$ is at least $|T|/6\cdot\delta$.

Initially, set $p_{u,v}(0)=1$ if $u=v$ and $p_{u,v}(0)=0$ otherwise. For each iteration $t$, set
\begin{gather}
p_{u,v}(t+1)=\frac{2\delta-\deg_{M_t}(v)}{2\delta}p_{u,v}(t)+\sum_{v'\in T}\frac{w_{M_t}(v',v)}{2\delta}p_{u,v'}(t) . \label{eq:puv}
\end{gather}
Note that $p_{u,v}(t+1)$ is a convex combination of $p_{u,v'}(t)$ over all $v'\in T$. Since the entropy function $\Phi_u(t)$ is concave, we obtain $\Phi_u(t+1)\ge\Phi_u(t)$, which implies that $\Phi(t+1)\ge\Phi(t)$.

Given a partition $(C,\overline C)$ on iteration $t$, define $q_u(t)=\sum_{v\in\overline C}p_{u,v}(t)$, which represents the probability that the random walk starting at $u$ ends up in $\overline C$.
\begin{claim}
$\sum_{u\in C}q_u(t)<|T|/100$.
\end{claim}
\begin{proof}
We have 
\begin{align*}
\sum_{u\in C}q_u(t+1)&=\sum_{u\in C}\sum_{v\in\overline C}p_{u,v}(t+1)\\&=\sum_{u\in C}\sum_{v\in\overline C}\left(\frac{2\delta-\deg_{M_t}(v)}{2\delta}p_{u,v}(t)+\sum_{v'\in T}\frac{w_{M_t}(v',v)}{2\delta}p_{u,v'}(t)\right)
\\&=\sum_{u\in C}\sum_{v\in\overline C}\frac{2\delta-\deg_{M_t}(v)}{2\delta}p_{u,v}(t)+\sum_{u\in C}\sum_{v\in\overline C}\sum_{v'\in T}\frac{w_{M_t}(v',v)}{2\delta}p_{u,v'}(t)
\\&=\sum_{u\in C}\sum_{v\in\overline C}\frac{2\delta-\deg_{M_t}(v)}{2\delta}p_{u,v}(t)+\sum_{v\in\overline C}\sum_{v'\in C}\frac{w_{M_t}(v',v)}{2\delta}\sum_{u\in C}p_{u,v'}(t)
\\&=\sum_{u\in C}\sum_{v\in\overline C}\frac{2\delta-\deg_{M_t}(v)}{2\delta}p_{u,v}(t)+\sum_{v\in \overline C}\sum_{v'\in C}\frac{w_{M_t}(v',v)}{2\delta}
\\&\le\sum_{u\in C}\sum_{v\in\overline C}p_{u,v}(t)+\sum_{v\in \overline C}\sum_{v'\in C}\frac{w_{M_t}(v',v)}{2\delta}
\\&=q_u(t)+\frac{w_{M_t}(C,\overline C)}{2\delta},
\end{align*}
By induction on $t$, we obtain $\sum_{u\in C}q_u(t+1)=(\sum_{i=1}^tw_{M_t}(C,\overline C))/2\delta$. Note that $\sum_{i=1}^tw_{M_t}(C,\overline C)$ is the total value of the cut $(C,\overline C)$ in the cut-graph $H_t$, which has value at most $|T|\delta/50$ by the construction of cut $(C,\overline C)$. It follows that $\sum_{u\in C}q_u(t)<|T|/100$.
\end{proof}

Since $|C|\ge|T|/3$, the values $q_u(t)$ for $u\in C$ have average at most $3/100$. By Markov's inequality, a constant fraction have value $q_u(t)\le1/24$. We will show that for each vertex $u\in T$ with $q_u(t)\le1/24$, we have $\Phi_u(t+1)\ge\Phi_u(t)+\Omega(1)$. This would imply $\Phi(t+1)\ge\Phi(t)+\Omega(|T|)$ and finish the analysis.

For the rest of the proof, fix a vertex $u\in T$ with $q_u(t)=\sum_{v\in\overline C}p_{u,v}(t)\le1/24$. By Markov's inequality, at most $1/8$ fraction of the vertices in $\overline C$ have $p_{u,v}(t)\ge1/3$; call these vertices \emph{bad}. Similarly, $\sum_{v\in C}p_{u,v}(t)\ge23/24$, and by (reverse) Markov's inequality, at most $1/8$ fraction of the vertices in $C$ have $p_{u,v}(t)\le2/3$; call these vertices \emph{bad}. Overall, at most $|T|/8$ vertices are bad. Now consider the matching $M_{t+1}$ of total weight at least $|T|/6\cdot\delta$. Each vertex has degree at most $\delta$ in $M_{t+1}$, so at most $|T|/8\cdot\delta$ weight of edges in $M_{t+1}$ are incident to bad vertices. So a constant fraction of the edges of $M_{t+1}$ (by weight) have both endpoints good, which means one endpoint $u$ has value $p_{u,v}(t)\le1/3$ and the other has value at least $2/3$. The definition of $p_{u,v}(t+1)$ in (\ref{eq:puv}) will ``mix'' these separated values, and a tedious but straightforward algebraic calculation establishes $\Phi(t+1)\ge\Phi(t)+\Omega(|T|)$.

\section{Sparsification Procedure}\label{appendix:sparsification}
The details of the full sparsification procedure from \cref{sec:sparsification} are detailed here. We prove the two cases separately:
  
\subsection{Unbalanced Case}     
We use the following method from \cite{steiner-polylog-flows} to deal with the unbalanced case:
\begin{lemma}[Theorem 4.2 from \cite{steiner-polylog-flows}]\label{lemma:unbalanced_case}
    Consider a graph $G = (V, E)$, a parameter $k\geq 1$, and a $k$-unbalanced set $U \subseteq T$. Then, we can compute the minimum Steiner cut of $G$ in $k^{O(1)}\text{polylog}(n)$ many $s-t$ max-flow computations plus $\Tilde{O}(m)$ deterministic time.
\end{lemma}
With $k\gets C\log^Cn$, this gives us polylogarithmic maximum flow calls and near-linear additional runtime as desired.

\subsection{Balanced Case}   
If $U$ is $k$-balanced, we use a sparsification procedure by using the $(s, \delta, \gamma)$-terminal-strong decomposition with terminal set $U$ to find a subset $U' \subset U$ such that $|U'| \leq |U|/2$ with the guarantee that some Steiner minimum cut contains at least one vertex from $U'$ in both of its sides. We then set $U \gets U'$.
   
We define a cluster $V_i$ to be \emph{trivial} if $|U_i| = 0$, \emph{small} if $1\leq |U_i| \leq s^2$, and \emph{large} if $|U_i| > s^2$. To construct set $U'$, for each cluster $V_i$, we take an arbitrary vertex from $U_i$ if $V_i$ is small, or $s+1$ arbitrary vertices from $U_i$ if $V_i$ is large. To prove correctness, we show that $U'$ is always at least a constant factor smaller than $U$ each iteration, and that $U'$ always contains at least one terminal on both sides of a minimum Steiner cut if $U$ is $k$-balanced.

\subsubsection{Size Bound}
\begin{claim}\label{claim:total_clusters}
    There are at most $O(\psi\cdot|U|\log n)$ total clusters, i.e. $\ell \leq O(\psi\cdot|U|\log n)$.
\end{claim}
\begin{proof}
    The total weight of intercluster edges is upper-bounded by $O(\psi\cdot\delta\cdot|U|\log n)$ from the guarantee of our $(s, \delta, \gamma, U)$-terminal-strong decomposition. We set $\delta \gets \Tilde{\lambda}$, where $\Tilde{\lambda}\in[\lambda,2\lambda]$ denotes a 2-approximation of the value of the minimum Steiner cut $\lambda$ on graph $G$. Since $\partial V_i$ is a Steiner cut in graph $G$, we have $\lambda \leq w(\partial V_i)$. Therefore
    \begin{equation}
        \ell\lambda \leq \sum_{i\in[\ell]} w(\partial V_i) \leq  O(\psi\cdot\lambda\cdot|U|\log n)
    \end{equation}
    Dividing by $\lambda$ on the left and right sides gives us our claim.
\end{proof}

\begin{lemma}
    The sparsification procedure above returns a set $U'$ such that $|U'|\leq |U|/2$.
\end{lemma}
\begin{proof}
    We can only have less than $|U|/s^2$ large clusters, and at most $O(\psi\cdot|U|\log n)$ small clusters due to \cref{claim:total_clusters}. From our construction, $U'$ has total size
    \begin{equation}
        |U'| < |U|/s^2\cdot(1+s) + O(\psi\cdot|U|\log n) \cdot 1
    \end{equation}
    With $s=O((L_{\max}/\psi)^2\log^2n) = \Tilde O(1/\psi^2)$ and $\psi = 1/\text{polylog}(n)$, we get that $|U'| \leq |U|/2$ as desired with a small enough chosen $\psi$.
\end{proof}

\subsubsection{Hitting Both Sides of the Minimum Steiner Cut}
This following claim ensures that only a few number of clusters are actually cut (have terminals on both sides) by the minimum Steiner cut.
\begin{claim}\label{claim:small_number_cut}
    (Analogous to Claim 4.12 in \cite{steiner-polylog-flows}) Let $C$ be one side of a minimum Steiner cut of $G$. Then, $C$ cuts at most $1/\gamma$ clusters of $G'$ (we define a cluster as being cut if there is at least one terminal on both sides of the cluster and both sides of the cut).
\end{claim}
\begin{proof}
    Steiner min-cut $\partial C$ has a maximum weight of $\delta$. For all clusters $V_i$, the portion of $C$ that intersects it (call this $\partial C_{V_i}$) has a terminal on both sides of it in $V_i$. From the definition of $(s, \delta, \gamma, U)$-terminal-strong, we must have that $w(\partial_{G[V_i]}C) \geq \gamma\cdot\delta$. Since clusters $G[V_i]$ are edge-disjoint and the minimum Steiner cut is upper-bounded by $\delta$, $\partial C$ cannot cut more than $1/\gamma$ clusters of $G'$.
\end{proof}

\begin{lemma}
    Suppose $U$ is $2s^2/\gamma$-balanced with witness $(S_1, S_2)$. Then $U' \cap S_i \neq \emptyset$ for both $i=1,2$.
\end{lemma}
\begin{proof}
The proof is a direct modification of Lemma 4.13 of \cite{steiner-polylog-flows}, replacing each instance of $1/\phi$ with either $s$ or $1/\gamma$. Call a cluster $V_i$:
 \begin{enumerate}
 \item \emph{white} if $S_1\cap U_i=\emptyset$ (i.e., $U_i\subseteq S_2$).
 \item \emph{light gray} if $0<|S_1\cap U_i|\le |S_2\cap U_i|<|U_i|$, which implies that $0<|S_1\cap U_i|\le s$.
 \item \emph{dark gray} if $0<|S_2\cap U_i|<|S_1\cap U_i|<|U_i|$, which implies that $0<|S_2\cap U_i|\le s$.
 \item \emph{black} if $S_2\cap U_i=\emptyset$ (i.e., $U_i\subseteq S_1$).
 \end{enumerate}
Every cluster must be one of the four colors, and by \Cref{claim:small_number_cut}, there are at most $1/\gamma$ many (light or dark) gray clusters since $U_i\cap S_1, U_i\cap S_2 \not= \emptyset$ implies that $S_1$ cuts cluster $V_i$. Note that since we are only considering clusters $V_i$ such that $U_i \not= \emptyset$, it must be that for a white cluster, we have $|S_2\cap U_i| \not= \emptyset$, and similarly, for a black cluster, we have $|S_1\cap U_i| \not= \emptyset$. There are now a few cases:
 \begin{enumerate}
 \item There are no large clusters. In this case, if there is at least one white and one black small cluster, then the vertices from these clusters added to $U'$ are in $S_2$ and $S_1$, respectively.
 Otherwise, assume w.l.o.g.\ that there are no black clusters. Since there are at most $1/\gamma$ gray clusters in total, $|S_1\cap U|\le 1/\gamma\cdot s^2$, contradicting our assumption that $\min\{|S_1\cap U|,|S_2\cap U|\}\ge 2s^2/\gamma$ for large enough $C$. 
 \item There are large clusters, but all of them are white or light gray. Let $V_i$ be a large white or light gray cluster. Since we select $s+1$ vertices of $U_i$, and $|S_1\cap U_i|=\min\{|S_1\cap U_i|,|S_2\cap U_i|\}\le s$, we must select at least one vertex not in $S_1$. Therefore, $S_2\cap U'\ne\emptyset$. If there is at least one black cluster, then the selected vertex in there is in $U'$, so $S_1\cap U'\ne\emptyset$ too, and we are done.
 
 So, assume that there is no black cluster. Since all large clusters are light gray (or white), $|S_1\cap U_i| \le s$ for all large clusters $V_i$. Moreover, by definition of small clusters, $|S_1\cap U_i| \leq |U_i| \le 1/s^2$ for all small clusters $V_i$. Since there are at most $1/\gamma$ gray clusters by \Cref{claim:small_number_cut},
\begin{align*}
|S_1\cap U| = \sum_{i: V_i\text{ small}}|S_1\cap U_i| + \sum_{i: V_i\text{ large}}|S_1\cap U_i| \le \frac1\gamma\cdot s^2 + \frac1\gamma\cdot s < \frac{2s^2}\gamma ,
\end{align*}
a contradiction.
\item There are large clusters, but all of them are black or dark gray. Symmetric case to (2) with $S_1$ replaced with $S_2$.
\item There is at least one black or dark gray large cluster $V_i$, and at least one white or light gray large cluster $V_j$. In this case, since we select $s+1$ vertices of $U_i$ and $|S_2\cap U_i|=\min\{|S_1\cap U_i|,|S_2\cap U_i||\}\le s$, we must select at least one vertex in $S_1$. Similarly, we must select at least one vertex in $U_j$ that is in $S_2$.\qedhere
 \end{enumerate}
\end{proof}

Since $s,1/\gamma\le\text{polylog}(n)$, we can set $C$ large enough in the statement of \cref{theorem:sparsification} so that $C\log^Cn\ge2s^2/\gamma$. This completes the proof of the balanced case for \cref{theorem:sparsification}.

\end{document}